\documentclass[aps, prd, preprint, groupedaddress]{revtex4}
\usepackage{mathrsfs}
\usepackage{bbm}
\usepackage{amsfonts}
\usepackage{amsmath}
\usepackage{graphicx}

\newtheorem{lemma}{Lemma}

\newtheorem{theorem}{Theorem}

\newtheorem{remark}{Remark}

\newenvironment{proof}{Proof:}

\def\endofproof {\hfill{$\Box$}\\}

\newcommand{\void}[1]{}

\newcommand{\hu}{\ensuremath{\widehat{\mathfrak{u}}}}
\newcommand{\nw}{\ensuremath{\,\mathfrak{nw}\,}}
\newcommand{\hnw}{\ensuremath{\,\widehat{\mathfrak{nw}}\,}}

\begin{document}

    \title{A no-ghost theorem for the bosonic Nappi-Witten string}
    
\author{
Gang~Chen,$^{1}$~~Yeuk--Kwan~E.~Cheung,$^{1}$~~Zheyong~Fan,$^{1}$\\
Jens~Fjelstad,$^{2}$~and~Stephen~Hwang$^{2}$
}
\affiliation{
$^{1}$Department of Physics, Nanjing University\\
22 Hankou Road, Nanjing 210098, China\\
$^{2}$Faculty of Technology and Science, Karlstad University\\
SE-651 88 Karlstad, Sweden
}

\hspace{1cm}
\begin{abstract}
We prove a no-ghost theorem for a bosonic string
propagating in Nappi-Witten spacetime.  This is  achieved in
two steps.  We first demonstrate unitarity
	for a class of
$NW/U(1)$ modules: the norm of any state which is primary with respect
to a chosen timelike $U(1)$ is non-negative. 
We then show that physical states--states satisfying the Virasoro constraints--in a class of modules of an affinisation of the Nappi-Witten algebra  are contained in the $NW/U(1)$ modules.
Similar to the case of  strings  on $AdS_3$, in order to saturate the spectrum obtained in light-cone quantization we are led to include modules with energy not bounded from below, which are related to modules with energy bounded from below by spectral flow automorphisms.
	
\end{abstract}

% insert suggested PACS numbers in braces on next line

\date{\today}
\maketitle

\section{Introduction}

For a string propagating in a target space-time $\mathrm{M}$ with metric of Lorentzian signature, unitarity of the space of physical states has been shown only for relatively few choices of $\mathrm{M}$ beyond the flat case and $\mathrm{dim}(\mathrm{M})>2$. The first non-trivial case to be investigated was $\mathrm{M}=SL(2,\mathbb{R})$, where the string in conformal gauge is described by the WZW model on the group $SL(2,\mathbb{R})$. The analysis of the $SL(2,\mathbb{R})$ WZW model has so far spanned  two decades.  And although the structure of the conformal field theory is still not completely understood, the question of unitarity of the corresponding string is by now settled \cite{Hwang, Hwang2, Evans, MalOog1, MalOog2}.
Beyond the $SL(2,\mathbb{R})$-string there are some scattered results in the literature, such as a class of non-compact coset models \cite{MMS:geom, BjHw1, BjHw2}.

We add another example to this list by considering a string on the Nappi-Witten group, $NW$, a non-semisimple group with dimension $4$. It was first shown in \cite{NappiWitten} that although the group is not semisimple, there exist a non-degenerate invariant bilinear form on its Lie algebra $\nw$~\footnote{There seems to be no unversally accepted notion of the Nappi-Witten group, in this paper we will take it to be the unique connected and simply connected group with the given Lie algebra} (in fact there is, up to an overall normalisation, a one-parameter family of such forms), and there is a WZW model defined on it. An underlying assumption is that the space of states of the WZW model can be fully decomposed into certain types of representations of a centrally extended loop algebra $\hnw$ of $\nw$, namely prolongation modules 
	of unitary simple modules, and modules related to these by spectral flow automorphisms of $\hnw$. 
	We show that the space of physical string states is unitary for 
	a certain subset of these modules.
	The allowed representations include subsets of the unitary representations of the highest and lowest weight, all representations of the continuous series, the trivial representation, and all representations obtained from these by spectral flow. The proof of the no-ghost theorem follows closely the proof for the $SL(2,\mathbb{R})$ string, and there are many similarities between the two theories. 
	
In a broader context of interest the Nappi-Witten model appears as an example of an exactly solvable string theory on the maximally symmetric backgrounds of plane-polarized gravitational waves~\cite{metsaev, Blau:2001ne, Berenstein:2002jq, Russo:2002rq}.
Classically the solvability comes with the realization that such models become free in the light-cone gauge.  Nappi-Witten model, on the other hand, has the added merit of being exactly solvable at the quantum level as  a WZW model~\cite{NappiWitten}.  In fact  interactions of closed strings in such a background has been analyzed with in a covariant manner~\cite{Cheung}.  In particular correlation functions of an arbitrary number of scattering closed strings are given, utilizing a set of free-field realization of the algebra introduced also in the same paper.  Nappi-Witten space being the closest analogue of four-dimensional Minkowski spacetime yet encompassing the three-form gauge field is sure to provide more interesting playgrounds for connecting string theory with observations.

The paper is organized as follows. In section~\ref{sec:coset} we show that the states in the relevant $\hnw$-representations which are primary with respect to a chosen timelike $\hu(1)$ (i.e., states in a certain $NW/U(1)$ state space) have non-negative norm. This section follows the main idea of the corresponding procedure for $SL(2,\mathbb{R})$ presented in~\cite{Dixon}, but requires more work due to the non-semisimplicity of $\nw$. Section~\ref{sec:noghost} contains the proof that physical states, i.e. states satisfying the Virasoro constraints, are contained in the $NW/U(1)$ state space, assuming that we complete our theory with an arbitrary unitary CFT to get total central charge $26$. This section follows closely \cite{Hwang2} for the un-flowed representations, and the proof goes through with minor modifications, while the extension of the proof to the spectral flowed representations becomes relatively straightforward and is very similar to the analogous proof for $SL(2,\mathbb{R})$ \cite{MalOog1}.

\section{Unitarity of $NW/U(1)$ modules}
\label{sec:coset}

When describing a critical bosonic string in terms of a CFT, the physical string states are the states in the CFT Hilbert space satisfying the Virasoro constraints:
\begin{eqnarray}
	L_n|\psi\rangle & = & 0\text{ for } n>0\label{vir1}\\
	\left(L_0-1\right)|\psi\rangle & = & 0\text{ (on-shell condition).}
	\label{vir2}
\end{eqnarray}
For target manifolds with metric of Lorentzian signature the CFT will generically be non-unitary.  And therefore one must show that the space of physical states is unitary, i.e. contains no states with negative norm.
The no-ghost theorem for flat (Minkowski) space was originally proved in the ``old covariant quantization'' scheme by Goddard and Thorn~\cite{Goddard}, see also extended discussion in~\cite{Thorn}.
For a group manifold $G$, more precisely for
a string based on a WZW model on a non-compact group $G$ with
either exactly one compact or exactly one non-compact generator,
the proof has two parts. Part I involves showing that states in
the coset $G/U(1)$ (i.e. states in the gauged $G/U(1)$ WZW model)
have non-negative norm if the $U(1)\subset G$ subgroup is chosen to span a
timelike curve in $G$.
  
    Recall that the space of states in a gauged $G/H$ WZW model is the subspace of the space of states of the $G$ WZW model spanned by highest weight states of $H$. In a $G/U(1)$ model we are thus considering vectors in representations of $\widehat{\mathfrak{g}}_k$, $\mathfrak{g}=\mathrm{Lie}(G)$, that are annihilated by all positive modes $J_n$, $n>0$ of the $U(1)$ current $J$.
    	One can only expect unitarity of the $G/U(1)$ state space if $G$ has at most one timelike direction. For the case of $G=SO(2,1)$ this was first established in \cite{Dixon},
using in particular a certain completeness relation. More precisely one uses the fact that a projection down to $U(1)$ primaries and descendents is the identity. In the Nappi-Witten group we first show this equality in detail, finding an explicit expression for an arbitrary state as a $U(1)$ descendent.
Part II of the proof is more formal,  
and one can partly use results from \cite{Goddard}.
This subsection is devoted to Part I and we leave Part II to the next
section.

\subsection{Current algebra in Nappi-Witten space}

The Nappi-Witten group $NW$ is four dimensional group obtained by  centrally extending  the two-dimensional Poincar\'e Group.  Its Lie algebra $\nw$ can be expressed in terms of the
anti-hermitian generators, $J^{+}, J^{-}, J, T$, with commutation relations:
\begin{eqnarray}
[J^{+},~ J^{-}]=2~ i~ T, ~~[J,~ J^{+}]=i
~J^{+},~~ [J,~ J^{-}]=-i~ J^{-}, ~~[T,
~Q^{\alpha}]=0~,
\end{eqnarray}
where we use $Q^{\alpha}$ to denote an arbitrary generator.  
One important aspect of $\nw$ is that it is not semisimple~\footnote{  
See 
\cite{Sfetsos:1993rh,  Sfetsos:1993na,  Olive:1993hk, Mohammedi:1993rg, FigueroaO'Farrill:1994yf}
for a general construction of WZW models with non-semisimple algebras. 
}
and consequently admits (up to normalization) two independent invariant bilinear forms which we call $\Omega$ and $\kappa$:
\begin{align}
\Omega=\left(
\begin{array}{cccc}
 0 & 2 & 0 & 0 \\
 2 & 0 & 0 & 0 \\
 0 & 0 & 0 & 1 \\
 0 & 0 & 1 & 0
\end{array}
\right)~~~~~~~~&
\kappa=\left(
\begin{array}{cccc}
 0 & 0 & 0 & 0 \\
 0 & 0 & 0 & 0 \\
 0 & 0 & -2 & 0 \\
 0 & 0 & 0 & 0
\end{array}
\right)~. 
\end{align}
We then obtain  a family of non-degenerate invariant bilinear forms, $\Omega_b= \Omega - \frac{b}{2} \kappa$.
It is easy to see that $J+\alpha T$ is a timelike vector w.r.t. $\Omega_b$ when $b+2\alpha<0$.
 As a result, we should prove that the highest weight states of some choice of $U(1)$ current
algebra generated by $J+\alpha T$ such that $b+2\alpha<0$, i.e. states $|\Psi\rangle$ such
that $(J_{n}+\alpha T_{n})|\Psi\rangle =0$ when $n>0$, have non-negative norm.

There are four types of unitary
representations, $V_+^{p^{+},p^{-}}$,
$V_{-}^{p^{+},p^{-}}$,
$V_{\alpha}^{0,p^{-}}$, $V^{0,0}$~\cite{KirKoun}, classified by the eigenvalues of the central
generator $T=ip^{+}$ and the quadratic Casimir
$\mathcal{C}^{(2)}$ (These are also nicely summarized in Appendix B1 of~\cite{Cheung}):
\begin{itemize}
\item
The representations $V_{+}^{p^+,p^-}$ are
highest weight representations (by which we here mean
representations generated from an eigenvector of $T$ and $J$,
annihilated by $J^+$) with highest weight vector $|p^+,p^-\rangle$
such that $T|p^+,p^-\rangle=ip^+|p^+,p^-\rangle$ and
$J|p^+,p^-\rangle=ip^-|p^+,p^-\rangle$. They are unitary for
$p^+>0$, $p^-\in\mathbb{R}$, and the quadratic Casimir has the
eigenvalue $\mathcal{C}^{(2)}=-2p^+(p^-+\frac{1}{2})$. Since this class of representations will play the
main role in the rest of the paper we provide a summary of their
structure.\\
We obtain bases $\{|r\rangle\}_{r\in\mathbb{N}_0}$ of
the vector spaces $V_+^{p^+,p^-}$
 by acting with $J^{-}$ on the vacua
$|0\rangle:=|p^{+},p^{-}\rangle$ for fixed $p^+$, $p^-$, which are annihilated by $J^{+}$.
\begin{eqnarray}
|r\rangle=(-iJ^{-})^{r}|0\rangle
\end{eqnarray}
 The generators act on the
representation as
\begin{eqnarray}
T|r\rangle &=&ip^{+}|r\rangle \\
J^{+}|r\rangle &=& 2ip^{+}r|r-1\rangle \\
J^{-}|r\rangle &=&i|r+1\rangle \\
J|r\rangle &=&i(p^{-}-r)|r\rangle.
\end{eqnarray}
There is an inner product on $V_+^{p^+,p^-}$ that takes the following form on the basis elements
\begin{eqnarray}
\langle s|r\rangle=(2p^{+})^{r}r! \delta_{r,s}.
\end{eqnarray}

\item
There are also lowest weight representations, $V_{-}^{p^+,p^-}$ generated from a lowest weight vector (annihilated by $J^-$) $|p^+,p^-\rangle_-$. These are unitary for $p^+<0$, $p^-\in\mathbb{R}$. It can be shown \cite{Cheung} that there exists a non-degenerate invariant pairing between $V_{+}^{p^+,p^-}$ and $V_{-}^{p^+,p^-}$, so the latter is equivalent to the contragredient of the former. $\mathcal{C}^{(2)}$ takes the value $-2p^+(p^--\frac{1}{2})$ on $V_{-}^{p^+,p^-}$.

\item
The representations $V^{0,p^-}_\alpha$ have neither highest nor lowest weight vectors, and are characterized by $p^-\in[0,1)$ and $\alpha\in\mathbb{R}$ such that not both are zero simultaneously, where the quadratic casimir takes the value $\mathcal{C}^{(2)}=-\alpha^2$. In analogy with $SL(2,\mathbb{R})$ we will refer to this class of representations as the continuous series.

\item
Finally, the representation $V^{0,0}$ is the trivial representation.
\end{itemize}
    
The Nappi-Witten currents obey the
operator product relations
\begin{eqnarray}
~J^{+}(z) J^{-}(w)&\sim& \frac{2iT(w)}{z-w} +\frac{k}{(z-w)^{2}}, \\
~J(z) J^{+}(w)&\sim& \frac{iJ^{+}(w)}{z-w},\\
~J(z) J^{-}(w)&\sim& \frac{-iJ^{-}(w)}{z-w},\\
~J(z) T(w)&\sim&
\frac{\frac{1}{2}k}{(z-w)^{2}},\\
~J(z) J(w)&\sim& \frac{\frac{1}{2}k
b}{(z-w)^{2}},
\end{eqnarray}
where we assume the central element $k\in\mathbb{R}$.
By standard manipulations, we can convert them
into commutation relations for the Laurent modes
\begin{eqnarray}
~[J^{+}_{m},J^{-}_{n}]&=&2iT_{m+n}+k m\delta_{m+n},\label{nwcomfirst}\\
~[J_{m},J^{+}_{n}]&=&iJ^{+}_{m+n},\\
~[J_{m},J^{-}_{n}]&=&-iJ^{-}_{m+n},\\
~[J_{m},T_{n}]&=&\frac{k}{2} m\delta_{m+n},\\
~[J_{m},J_{n}]&=&\frac{k b}{2} m\delta_{m+n}.\label{nwcomlast}
\end{eqnarray}

The OPE of the energy-momentum tensor $\mathcal{E}(z)$ with a primary operator $Q(w)$ of conformal dimension one is
\begin{eqnarray}
~Q(z) \mathcal {E}(w)&\sim&
\frac{Q(w)}{(z-w)^{2}},
\end{eqnarray}
so demanding that the currents are Virasoro primaries of conformal dimension one restricts the ``Sugawara form'' of the energy-momentum tensor to
\begin{equation}
\mathcal{E}(z)=\frac{1}{k} \left(JT(z)+TJ(z)+\frac{J^{+}J^{-}(z)
+J^{-}J^{+}(z)}{2}+(\frac{2}{k}-b)TT(z)\right).
\end{equation}
The commutation relations of Virasoro generators with current
modes are
\begin{eqnarray}
~[L_{m},Q^{\alpha}_{n}]&=&-n Q^{\alpha}_{m+n}.
\end{eqnarray}
A standard calculation confirms that the central charge takes the value $c=4$, independently of the values of $k$ and $b$.
It is straightforward to check that one can set the parameter $b$ in $\Omega_b$ to any real value by an automorphism of $\nw$, and we thus expect that different values of $b$ correspond to equivalent CFT's. See Remark \ref{rem:bvalues} at the end of section \ref{sec:cosetunitarity}.
For simplicity we will work with the case $b=2/k$, such that the energy-momentum tensor has no term $TT(z)$.  
For the rest of the paper the name $\hnw$ denotes solely the Lie algebra given in (\ref{nwcomfirst})--(\ref{nwcomlast}) with $b=2/k$.

\subsection{Constraints from unitarity at low excitations}\label{sec:lowNunitarity}
    The representations that we consider as candidates for the space of states of the Nappi-Witten WZW model are the prolongation modules $\widehat{V}_{+}^{p^+,p^-}$, $\widehat{V}_{-}^{p^+,p^-}$ and $\widehat{V}_\alpha^{0,p^-}$ of $V_{+}^{p^+,p^-}$, $V_{-}^{p^+,p^-}$ and $V_\alpha^{0,p^-}$, respectively (plus the spectral flow of these representations, see section III). These are the $\hnw$-modules obtained by the requirement that any vector in the (horizontal) $\nw$-module is annihilated by the positive modes of the currents, $J^{\pm}_n$, $J_n$, $T_n$, $n>0$. Note that, similar to the $SL(2,\mathbb{R})$-case, the prolongation modules constructed from $V_{+}^{p^+,p^-}$ and $V_{-}^{p^+,p^-}$
     for $p^+\neq 0$ (i.e. the highest resp. lowest weight modules) are in fact equivalent to the Verma modules of $\hnw$ with highest resp. lowest weights $(p^+,p^-)$.\\

Some interesting observations follow by analyzing the Virasoro condition $(L_0-1)|\psi\rangle=0$. Acting with $L_0$ on a state $|\psi\rangle$ at level (mode number) $N$ in $\widehat{V}_{+}^{p^+,p^-}$ gives
\begin{equation}
    L_0|\psi\rangle=\left(\frac{-2p^+(p^-+\frac{1}{2})}{k}+N\right)|\psi\rangle
\end{equation}
Applying the Virasoro condition $(L_0-1)|\psi\rangle=0$, we see that, since the only requirement on $p^-$ is that it is real, this imposes no obvious constraint on the allowed representations. An analogous conclusion follows from looking at a lowest weight representation.
If instead $|\psi\rangle\in\widehat{V}^{0,p^-}_\alpha$, the same Virasoro condition reads
\begin{equation}
    \left(\frac{-\alpha^2}{k}+N-1\right)|\psi\rangle=0.
\end{equation}

If $k>0$ we may potentially  find for any $N\in\mathbb{N}$ a representation with a physical state at that level.
 If $k<0$, however, the only representations where excited states may be physical are those where $\alpha=0$. Furthermore, these representations only allow $N=1$. It is easily seen that as long as $p^-\neq 0$ the only states in such a representation that satisfy the constraints $L_n|\psi\rangle=0$, $n>0$, are the states $T_{-1}|\chi\rangle$ where $|\chi\rangle\in V^{0,p^-}_0$. It is equally easy to see that all of these states have zero norm. If $\alpha\neq 0$ then only vectors in the horisontal representation may satisfy the Virasoro conditions, and these have positive norm. Note that these conclusions still hold if we include a contribution to $L_0$ from another unitary CFT since such contributions are necessarily positive. To conclude, we have shown that if the level $k$ is negative, all physical states in the modules $\widehat{V}^{0,p^-}_\alpha$ have non-negative norm.\\
    
Next, consider the norms of some of the first excited string states in $\widehat{V}_+^{p^+,p^-}$.
\begin{eqnarray}
\|-iJ^{-}_{-1}|r\rangle\|^{2}&=&-\langle
r|[J^{+}_{1},J^{-}_{-1}]+J^{-}_{-1}J^{+}_{1}|r\rangle
=(2p^{+}-k)\langle r|r \rangle\\
\|-iJ^{+}_{-1}|r\rangle\|^{2}&=&-\langle
r|[J^{-}_{1},J^{+}_{-1}]+J^{+}_{-1}J^{-}_{1}|r\rangle
=(-2p^{+}-k)\langle r|r \rangle\\
\|-i(J_{-1}+T_{-1})|r\rangle\|^{2}&=&-\langle
r|J_{1}J_{-1}+T_{1}J_{-1}+J_{1}T_{-1}|r\rangle
=(-k-1)\langle r|r \rangle\\
\|-i(J_{-1}-T_{-1})|r\rangle\|^{2}&=&-\langle
r|J_{1}J_{-1}-T_{1}J_{-1}-J_{1}T_{-1}|r\rangle =(k-1)\langle r|r
\rangle
\end{eqnarray}
The first observation is that there is no possible value of $k$ such that none of these states have negative norm. Since $p^+>0$, the first two states can have non-negative norm only if $k<0$ and $p^+\leq-\frac{k}{2}$. In fact these conditions are necessary since the state $J^+_{-1}|p^+,p^-\rangle$ is physical for suitable values of $p^{-}$. It then follows that the last two states both have negative norm if $k\in (0,-1)$, and if $k\leq -1$ only the last state has negative norm. It is then natural to guess that we should also require $k\leq-1$. Let us, however, investigate this a bit further, and temporarily let the parameter $b$ be arbitrary. Consider the combination $R_{-1}=J_{-1}+\alpha T_{-1}$ for some $\alpha\in \mathbb{R}$, we then have
$$\| R_{-1}|r\rangle\|^2=-\frac{k}{2}(b+2\alpha)\langle r|r\rangle.$$
It follows that for any $b\in\mathbb{R}$ there are combinations of the currents $J$ and $T$ producing positive norm states, and other producing negative norm states. The picture in the $J-T$ plane with a metric given by $\Omega_{b}$ is as follows. For $b=0$ the lightcone is aligned with the $J$ and $T$ axes. For $b>0$ the lightcone is compressed, and when $b>2$ both combinations $J\pm T$ are spacelike. Similarily when $b<0$ the lightcone is widened, and for $b<-2$ both combinations $J\pm T$ are timelike. For $b=\pm 2$ the lighcone is aligned with the $T$ and $J\mp T$ axes. With our choice $b=2/k$ we thus get that both combinations $J\pm T$ are timelike if $k\in (-1,0)$, but there are still other combinations that are spacelike.  Demanding unitarity at the first excited level thus only gives the necessary conditions
\begin{eqnarray}
p^{+}&\leq &-\frac{k}{2}\label{uncond1}\\
k&<&0.\label{uncond2}
\end{eqnarray}
		
\subsection{Unitarity of $NW/U(1)$ states}
\label{sec:cosetunitarity}

We will now show that the conditions (\ref{uncond1}) and (\ref{uncond2}) are sufficient to ensure the absence of negative norm states in the coset module built from $\widehat{V}_+^{p^+,p^-}$ or $\widehat{V}_{-}^{p^+,p^-}$, i.e. the subspace spanned by states which are primary with respect to the timelike $\hu(1)$ current algebra generated by the modes $\{J_n-T_n\}_{n\in \mathbb{Z}}$. For brevity we only consider the case of $\widehat{V}_+^{p^+,p^-}$, the same result follows with obvious modifications for $\widehat{V}_{-}^{p^+,p^-}$.

\noindent
Let $|N; r, p^{+}, p^{-}\rangle$
be a general state, where $N$ is the excited
string level
and $r$ is defined by
$J_0|N;r,p^+,p^-\rangle=i(p^--r)|N;r,p^+,p^-\rangle$.
Such a state  obeys
\begin{equation}
\langle N; r, p^{+}, p^{-}|L_{0}|N; r, p^{+},
p^{-}\rangle
=(-\frac{2p^{+}(p^{-}+\frac{1}{2})}{k}+N)\langle
N; r, p^{+}, p^{-}|N; r, p^{+}, p^{-}\rangle.
\label{L01}
\end{equation}
Since
\begin{eqnarray}
L_{0}
&=&\frac{1}{k}[\frac{1}{2}(J^{+}_{0}J^{-}_{0}
+J^{-}_{0}J^{+}_{0})+(J_{0}T_{0}+T_{0}J_{0}){}\nonumber\\
&&{}+\sum_{n\geq 1}(J^{+}_{-n}J^{-}_{n}
+J^{-}_{-n}J^{+}_{n}+2 J_{-n}T_{n}+2
T_{-n}J_{n})], \end{eqnarray} we find that
\begin{eqnarray}
\langle N; r|L_{0}|N; r \rangle
&=&\frac{1}{k}\langle N; r|J^{-}_{0}J^{+}_{0}
+\sum_{n\geq 1}(J^{+}_{-n}J^{-}_{n}
+J^{-}_{-n}J^{+}_{n}+2 J_{-n}T_{n}+2
T_{-n}J_{n})|N; r\rangle {}\nonumber\\
&&{}+\frac{1}{k}(-p^{+}+2p^{+}r-2p^{+}p^{-})\langle N; r|N;
r\rangle. \label{L02}
\end{eqnarray}
From (\ref{L01}) and (\ref{L02}) we get
\begin{eqnarray}
\langle N; r|N; r\rangle=\frac{\langle N; r|J^{-}_{0}J^{+}_{0}
+\sum_{n\geq 1} (J^{+}_{-n}J^{-}_{n} +J^{-}_{-n}J^{+}_{n}+2
J_{-n}T_{n}+2 T_{-n}J_{n})|N; r\rangle} {N
k-2p^{+}r}\label{hwnorm}
\end{eqnarray}
As follows from (\ref{uncond1}), (\ref{uncond2}), and $r\geq -N$,
the denominator is manifestly negative for
$k<0$. We prove by induction that the numerator is also
negative.

The inductive step amounts to showing that the numerator is negative for given values of $N$ and $r$ if it is negative for all excitation levels less than $N$ and $J_0$ eigenvalues smaller than $r$ (note that $r$ is bounded from below by $-N$). For $N=0$ all states are $U(1)$-primaries, and those states form a unitary $\nw$-module. The only state with $r=-N$ is $(J^+_{-1})^N|0;0,p^+,p^-\rangle$, and it is easily checked that this state is a $U(1)$ primary, and that it has positive norm.
The proof of the inductive step follows the procedure described in \cite{Dixon}, where it was applied to $SL(2,\mathbb{R})$. The idea is to express the norm of a state of the form $J^a_{n}|N;r\rangle$, $n>0$, in terms of norms of states with lower $N$ and larger $r$. In the case of $SL(2,\mathbb{R})$ it follows easily since any state is a linear combination of descendents of primaries of the timelike $\hu(1)$. In our case, since $\nw$ is not semisimple, this does not follow trivially. As a first step we therefore show explicitly that any state in $\widehat{V}_{+}^{p^+,p^-}$ can itself be written as a combination of $\hu(1)$-descendent states. The same statement holds, with a similar proof, for $\widehat{V}_{-}^{p^+,p^-}$.

\begin{lemma}
\label{lem:completeness}

Any state $|\Psi\rangle$ in $\widehat{V}_{+}^{p^+,p^-}$ can be written
\begin{eqnarray}
|\Psi\rangle &=& |h^0\rangle
+\sum_{n>0}(J_{-n}-T_{-n})|h^1_{n}\rangle
\nonumber\\&+&\sum_{n_{1},n_{2}>0}
(J_{-n_{1}}-T_{-n_{1}})(J_{-n_{2}}-T_{-n_{2}})
|h^2_{n_{1},n_{2}}\rangle +\ldots\label{expansion}
\end{eqnarray}
where the states $|h^i_{n_1,\ldots,n_i}\rangle$ are annihilated by any $J_n-T_n$ for $n>0$, and only finitely many terms on the right hand side are non-zero.
\end{lemma}

\begin{proof}
    We first show that the Lemma follows if we can
    show that any state obtained by acting with one
    of the operators
    $Q_{-m}\in\{J_{-m}^{+},J_{-m}^{-},J_{-m},T_{-m}\}$,
    where $m>0$, on a $\hu(1)$-primary
     state can be written on the required form. Any state
   at level $N=0$ is already a
   $\hu(1)$ primary state. Any state
   at level $N=1$ is obtained as a linear combination of
   states of the form $Q_{-1}|h\rangle$ where $|h\rangle$
   is a state at level $N=0$. The statement follows by
   induction since any state at level $N_0$ is a linear
   combination of states of the form
   $Q_{-m}|\Phi_N\rangle$, and if we assume that we can
   write any state at level $N<N_0$ on the required form,
   then we have
    $$Q_{-m}|\Phi_N\rangle=Q_{-m}|h^0\rangle+
    \sum_{n>0}Q_{-m}(J_{-n}-T_{-n})|h^1_n\rangle
    +\ldots$$
Using the commutation relations of $\hnw$
we can re-write this expression as a sum of terms where
every term consists of some factors of negative
$\hu(1)$-modes acting on a state
$Q_{-s}|h\rangle$ for some $s>0$, where $|h\rangle$
is a $\hu(1)$-primary state, and the
induction is complete.

\noindent
We must now show that any state of the form $Q_{-m}|h\rangle$,
where $|h\rangle$ is a primary state, can be written on
the required form.
For $Q=J$ or $T$ this follows trivially. Note that the combination
$J_{-m}+(1-\frac{2}{k})T_{-m}$ commutes with any timelike
generator $J_n-T_n$. The decompositions
\begin{eqnarray}
J_{-m}&=&\frac{k}{2(k-1)}(J_{-m}+(1-\frac{2}{k})T_{-m})
+\frac{k-2}{2(k-1)}(J_{-m}-T_{-m}),\label{Jdec}\\
T_{-m}&=&\frac{k}{2k-2}
(J_{-m}+(1-\frac{2}{k})T_{-m})
-\frac{k}{2k-2}(J_{-m}-T_{-m})\label{Tdec}
\end{eqnarray}
thus immediately put $Q_{-m}|h\rangle$ on the required
form. It remains to investigate when $Q=J^\pm$. For
simplicity we restrict to $J^-$, the other case follows
analogously. 
Our strategy is to recursively construct the first term on the right hand side of (\ref{expansion}), and the result will be of a form where the Lemma becomes obvious.
Consider therefore a state
$J^-_{-m}|h_N\rangle$ with $|h_N\rangle$
a $\hu(1)$ primary state at level $N$.
Note that $J^-_N|h_N\rangle$ is either zero or a
$\hu(1)$ primary state, and that the
state $|h_{(1)}\rangle$ defined as
$$|h_{(1)}\rangle:= \left(J^-_{N-1}+\frac{2i}{k}T_{-1}J^-_N\right)|h_N\rangle$$ is a primary state. Thus, using (\ref{Tdec}) we see that $J^-_{N-1}|h_N\rangle=|h_{(1)}\rangle-\frac{2i}{k}T_{-1}J^-_N|h_N\rangle$ has the required form. Similarly, the state $|h_{(2)}\rangle$ defined by
$$|h_{(2)}\rangle:=\left(J^-_{N-2}+\frac{2i}{k}T_{-1}J^-_{N-1}
+\frac{1}{2}\left(\frac{2i}{k}\right)^2(T_{-1})^2J^-_N+
\frac{2i}{2k}T_{-2}J^-_N\right)|h_N\rangle$$
is primary, so using (\ref{Tdec}) and the decomposition of
$J^-_{N-1}|h_N\rangle$ we see that  $J^-_{N-2}|h_N\rangle$ has the
required form. Recursively one can then show that any state
$J^-_{-m}|h_N\rangle$ can be written on the required form by
constructing a primary state $|h_{(m+N)}\rangle$ given by the
general expression:
\begin{eqnarray}   \label{genprim}
  \displaystyle{
    |h_{(m+N)}\rangle }
    & := &  \displaystyle{ \sum_{n_1,n_2,\ldots,n_{m+N}\geq 0}\frac{1}{n_1!\cdots n_{m+N}!}\left(\frac{2i}{k}\right)^{n_1}\left(\frac{2i}{2k}\right)^{n_2}\cdots\left(\frac{2i}{(m+N)k}\right)^{n_{m+N}}
         }
     \nonumber \\
    & \times &   
    \displaystyle{ \left(T_{-1}\right)^{n_1}\cdots\left(T_{-m-N}\right)^{n_{m+N}}J^-_{-m+n_1+2n_2+\ldots+(m+N)n_{m+N}}|h_N\rangle~. }
   \end{eqnarray}
There are only a finite number of non-zero terms on the right hand side since $N<\infty$.
The first term in this expression is $J_{-m}^-|h_N\rangle$, and using recursively the same expression together with (\ref{Tdec}) on every other term one obtains the required form for $J_{-m}^-|h_N\rangle$.
By changing every factor $(2i)^{n_j}$ to $(-2i)^{n_j}$, and $J^-_{\ldots}$ to $J^+_{\ldots}$ in (\ref{genprim}) one obtains the analogous formula for $J^+_{-m}|h_N\rangle$.
This completes the proof of the Lemma.\endofproof
\end{proof}

Let $\mathcal{H}^{p^+,p^-}\subset\widehat{V}^{p^+,p^-}_+$ denote the subspace spanned by $\hu(1)$-primary states. Choose a linear projector $\mathbf{P}$ on $\widehat{V}^{p^+,p^-}_+$ down to $\mathcal{H}^{p^+,p^-}$, implying that $(J_n-T_n)\mathbf{P}=0$ for all $n>0$. It is clear from Lemma \ref{lem:completeness} that we can choose $\mathbf{P}$ such that also $\mathbf{P}(J_{-n}-T_{-n})=0$ for all $n>0$, and one can show that this implies $\mathbf{P}^\dagger=\mathbf{P}$.
Using Lemma~\ref{lem:completeness} it is then not difficult to show the completeness relation
\begin{eqnarray}  
\label{compl}
\mathrm{id}_{\widetilde{V}^{p^+,p^-}}
&=&\textbf{P}+\sum_{n>0}\frac{1}{(-(k-1) n)} X_{-n}\textbf{P}X_n{}\nonumber \\ 
&&+\frac{1}{2!}\sum_{n_{1},n_{2}>0} \frac{1}{(-(k-1)
n_{1})}\frac{1}{(-(k-1)
n_{2})}X_{-n_1}X_{-n_2}\textbf{P} X_{n_1}X_{n_2}
{}\nonumber\\ 
&&{} +\cdots
\end{eqnarray}
where $X_n:=J_n-T_n$.
The idea is now to insert (\ref{compl}) into each term on the right hand side of (\ref{hwnorm}), affording a sum over intermediate states.
Consider first the effect of inserting (\ref{compl}) in the middle of the term $\langle N,r|J^+_{-p}J^-_p|N,r\rangle$. After commuting the factors of $(J-T)$ to the left
and right, a typical term in the resulting sum is
\begin{eqnarray}
\langle N, r|[\frac{1}{m!}(k-1)^{-m}
\frac{1}{n_{1}\cdots n_{m}}
J^{+}_{-(p+n_{1}+\cdots+n_{m})} \textbf{P}
J^{-}_{(p+n_{1}+\cdots+n_{m})}]|N, r\rangle
\end{eqnarray}
We define $P=p+n_{1}+\cdots+n_{m}$. The properties of $\mathbf{P}$ implies $$\langle N,r|J^+_{-P}\mathbf{P}J^-_{P}|N,r\rangle=-\|\mathbf{P}J^-_P |N,r\rangle\|^2,$$
where $J^-_P|N,r\rangle$ has level $N-P<N$ if $P>0$. If every primary state of level strictly lower than $N$ has non-negative norm it follows that $\langle N,r|J^+_{-P}\mathbf{P}J^-_P|N,r\rangle\leq 0$. The same result follows by inserting (\ref{compl}) in the expression $\langle N,r|J^-_{-p}J^+_p|N,r\rangle$ for $p>0$.
Inserting the completeness relation in the term $\langle N,r|J^-_0J^+_0|N,r\rangle$ gives a sum of terms with the same property as above, plus the term $-\|\mathbf{P}J^+_0|N,r\rangle\|^2$. The state $J^+_0|N,r\rangle$ also has level $N$, but it has lower $J_0$-eigenvalue, and we can assume that all primaries at level $N$ but with $J_0$-eigenvalue lower than $r$ have non-negative norm it follows that $\langle N,r|J^-_0\mathbf{P}J^+_0|N,r\rangle\leq 0$.
Finally, using (\ref{Jdec}) and (\ref{Tdec}) and the notation $Y_n=J_n+(1-\frac{2}{k})T_n$, it follows that
$$\langle N,r| 2T_{-n}J_n+2J_{-n}T_n|N,r\rangle=\left(1-\frac{1}{k}\right)^{-1}\langle N,r|Y_{-n}Y_n|N,r\rangle,$$
where $Y_n|N,r\rangle$ is again a $\hu(1)$ primary state.
Summing up the terms in (\ref{hwnorm}), we find that the
operators between in-state and out-state are of
the form
\begin{eqnarray}
{\bf\mathcal{O}} 
&=&\sum_{P\geq 1} \sum_{m\geq0}\sum_{n_{1}\geq0,\cdots, n_{m}\geq0}^{n_{1}+\cdots +n_{m}\leq P-1}
\frac{1}{m!}(k-1)^{-m} \frac{1}{n_{1}\cdots
n_{m}} J^{+}_{-P} \textbf{P} J^{-}_{P}   \nonumber
\\ 
&&+ ~\sum_{P\geq 0} \sum_{m\geq 0} \sum_{n_{1}\geq0,\cdots, n_{m}\geq0}^{n_{1}+\cdots +n_{m}\leq P}
\frac{1}{m!}(k-1)^{-m} \frac{1}{n_{1}\cdots n_{m}} J^{-}_{-P} \textbf{P} J^{+}_{P}
 \nonumber
 \\ 
&&+~\left(1-\frac{1}{k}\right)^{-1}\sum_{n\geq 1}Y_{-n}Y_n\nonumber \\
&=&\sum_{P\geq 1}F_{P}(k)J_{-P}^{+}\textbf{P} J_{P}^{-} + \sum_{P\geq0}F_{P-1}(k)J_{-P}^{-}\textbf{P} J_{P}^{+} +
\left(1-\frac{1}{k}\right)^{-1}\sum_{n\geq 1}Y_{-n}Y_n~,
\end{eqnarray}
where $F_{P}(k)=(1+\frac{1}{(k-1)})(1+\frac{1}{2(k-1)})\cdots (1+\frac{1}{P (k-1)})$ for $P>0$ and $F_0(k)=1$. 
In particular $F_P(k)>0$ for $k<0$, as well as $(1-1/k)>0$.
Since all states at $N=0$ as well as the $\hu(1)$-primary states with $r=-N$ have non-negative norm, it follows by induction that the numerator of (\ref{hwnorm}) is negative for all $N\in \mathbb{N}_0$ and $r\geq -N$.
Hence, we conclude that there are
no negative norm states in $\mathcal{H}^{p^+,p^-}\subset\widehat{V}_+^{p^+,p^-}$ when $p^{+}\leq
-\frac{k}{2}$ and $k<0$. The same result holds in the case of $\widehat{V}_-^{p^+,p^-}$ for $\frac{k}{2}\leq p^+$. 
This completes the
first part of the proof of our no-ghost theorem for bosonic strings on the
Nappi-Witten group.
\begin{remark}\label{rem:bvalues}
Denote by $\hnw_b$ the algebra (\ref{nwcomfirst})--(\ref{nwcomlast}) with an arbitrary $b\in\mathbb{R}$. There is an automorphism $J_0\mapsto J_0+\mu T_0$ of the horisontal subalgebra that  shifts the value of b according to $b\mapsto b+2\mu$, which we took as a sign that different values of $b$ correspond to equivalent CFT's. The argument can be made stronger by considering a change of basis of $\hnw_b$ where $J_n$ is replaced with $\widetilde J_n=J_n+\mu T_n$. This is not an automorphism, but changes only the commutator (\ref{nwcomlast}) to
$$[\widetilde J_m,\widetilde J_n]=\frac{1}{2}k(b+2\mu)m\delta_{m+n}.$$
In other words, we have shown that $\hnw_b\cong\hnw_{b+2\mu}$ for any $\mu\in\mathbb{R}$.
When discussing highest weight modules we have used the triangular decomposition $\hnw_b=\hnw_b^{-}\oplus\hnw_b^{0}\oplus\hnw_b^{+}$ where all positive (negative) modes are in the $+$ ($-$) parts, and in addition $J^-_0\in\hnw_b^{-}$, $J^+_0\in\hnw_b^+$. Note that the change of basis above preserves the triangular decomposition. This implies in particular that the $\hnw_b$-module $\widehat{V}_\pm^{p^+,p^-}$ is the same thing as the $\hnw_{b+2\mu}$-module $\widehat{V}_\pm^{p^+,p^-+\mu p^+}$. Using this fact it is not difficult to check explicitly that for any $b=\frac{2}{k}+2\mu$ such that the vector $J-T$ is still timelike, the proof of unitarity above still holds. For other values of $b$ one must choose a different $\hu(1)$ subalgebra, but this is straightforward. With the same constraint on $b$ one can show that also the results of the next section holds, and we expect similarily that this extends straightforwardly to arbitrary values of $b$.
\end{remark}

\section{No-ghost theorem for Nappi-Witten string theories}\label{sec:noghost}
The idea of the proof of the no-ghost theorem follows \cite{Hwang2}, see also \cite{Evans}, using also classic results by Goddard and Thorn \cite{Goddard}.

\noindent
Denote by $\widehat{V}_{+(N)}^{p^+,p^-}$ the subspace of $\widehat{V}_+^{p^+,p^-}$ spanned by states of excitation level $N$ or less. Furthermore, let $\mathcal{T}\subset\mathcal{H}^{p^+,p^-}$ be the subspace of $\hu(1)$ primary states that are also Virasoro primary, i.e. annihilated by $L_n$ with $n>0$. Use the notation $X_n:=J_n-T_n$ for the generators of the timelike $\hu(1)$. For the rest of this section it is important that we complete the Nappi-Witten CFT with a unitary CFT with $c=22$ such that the total central charge takes the value $26$.
\begin{lemma}\label{lem:basis}
The set of states of the form
\begin{equation}
|\{\lambda,\mu\},\psi\rangle:=L_{-1}^{\lambda_1}\cdots L_{-p}^{\lambda_p}X_{-1}^{\mu_1}\cdots X_{-p}^{\mu_p}|\psi\rangle
\label{basis}
\end{equation}
where $|\psi\rangle\in\mathcal{T}$ has level $N_\psi\leq N$ and $\sum_{a=1}^p a(\lambda_a+\mu_a)+N_\psi=N$ for a fixed $N$, form a basis for $\widehat{V}_{+(N)}^{p^+,p^-}$
	if $p^+\in (0,-\frac{k}{2})$.
\end{lemma}

\begin{proof}
We first prove that states of the form (\ref{basis}) are linearly independent, and then proceed by showing that they span $\widehat{V}_{+(N)}^{p^+,p^-}$.
The field $\mathcal{E}^X(z)=\frac{1}{2(1-k)}XX(z)$ defines an energy-momentum tensor for a $U(1)$ ($c=1$) theory with current $X(z)$. The Laurent modes $L^X_n$ are given by $L^X_n=\frac{1}{2(1-k)}\sum_{m\in\mathbb{Z}}:X_mX_{n-m}:$.
The conventional coset construction gives a representation of the Virasoro algebra with central charge $c=25$ generated by $L^{c}_n=L_n-L^X_n$ since $L^c_m$ commutes with $X_n$. Since $L^X_n$ are bilinears in modes $X_m$ we might as well switch the generators $L_{-n}$ on the right hand side of (\ref{basis}) to the generators $L^c_{-n}$. Since $L^c_n$ and $X_m$ commute, the states of this form span a highest weight representation of $\mathrm{Vir}\oplus\hu(1)$ with $c=25$. The $L^c_0$-eigenvalues on $\widehat{V}_+^{p^+,p^-}$ are of the form $h^c=-\frac{2p^+(p^-+\frac{1}{2})}{k}+\frac{(p^++r-p^-)^2}{2(1-k)}+N$ with $N\in\mathbb{N}_0$ and $r\in\mathbb{Z}$ such that $r\geq -N$. We immediately see that if $p^-+\frac{1}{2}\geq 0$ then $h^c\geq 0$ (for $k<0$). One checks that with the constraints $k<0$, $p^+\leq-k/2$, $N\geq 0$, $r\geq -N$, $h^c$ as a function of $p^-$ can only change sign if $p^+=-k/2$ (and then only if $r=-N$). In other words, for $0<p^+<-k/2$ the values of $h^c$ are positive.
The Kac determinant of the representation of $\mathrm{Vir}\oplus\hu(1)$ is the product of Kac determinants of the $\hu(1)$ part and the Virasoro part. Since the Kac determinant of the $\hu(1)$ part is always non-zero, and the Virasoro contribution for $c=25$ vanishes only for $h^c\leq0$, it follows that states of the form (\ref{basis}) are indeed linearly independent.

Let $N$ be the total level of a state of the form (\ref{basis}), i.e. there are contributions $\sum_a a(\lambda_a+\mu_a)$ from the oscillators, and $N_\psi$ from $|\psi\rangle$ adding up to $N$. The states of the form (\ref{basis}) with $N=0$ obviously span $\widehat{V}_{+(0)}^{p^+,p^-}$. Assume that the states at level $N-1$ or less span $\widehat{V}_{+(N-1)}^{p^+,p^-}$. Let $B_{(N)}$ denote the linear span of states of the form (\ref{basis}) of level $N$ or less and where $N_\psi<N$. Since the Kac determinant is non-zero there are no null states in $B_{(N)}$, so $\widehat{V}_{+(N)}^{p^+,p^-}\cong B_{(N)}\oplus B_{(N)}^\perp$ where $B_{(N)}^\perp$ is the orthogonal complement of $B_{(N)}$ in $\widehat{V}_{+(N)}^{p^+,p^-}$. Pick an arbitrary $|u\rangle\in B_{(N)}^\perp$. Then for any $|v_{s}\rangle\in V_{+(N-s)}^{p^+,p^-}$, $s>0$, we have
$$0=\left(|u\rangle,L_{-s}|v_s\rangle\right)=\left(L_s|u\rangle,|v_s\rangle\right).$$
Since by assumption the inner product is non-degenerate on
$\widehat{V}_{+(N-s)}^{p^+,p^-}$ for $0<s\leq N$ we conclude that
$L_s|u\rangle=0$ for any $s\in [1,N]$. We also have trivially that
$L_n|u\rangle=0$ for $n>N$ so we conclude that $L_n|u\rangle=0$
for any $n>0$. One shows similarly that $X_n|u\rangle=0$ for any
$n>0$. We have thus shown that
$B_{(N)}^\perp=\mathcal{T}\cap\widehat{V}_{+(N)}^{p^+,p^-}$, and
the lemma follows.\endofproof
\end{proof}

A {\em spurious state} is a linear combination of states of the form (\ref{basis}) with $\lambda\neq 0$.
Recall the result of Goddard and Thorn \cite{Goddard} that when $c=26$, the space of on-shell spurious states (i.e. spurious states $|s\rangle$ satisfying $(L_0-1)|s\rangle=0$) is closed under the action of $L_n$ with $n>0$.
It follows that if $|\psi\rangle=|\phi\rangle+|s\rangle$ is on shell with $|s\rangle$ spurious and on-shell, the condition $L_n|\psi\rangle=0$ for $n>0$ is equivalent to $L_n|\phi\rangle=0$ and $L_n|s\rangle=0$.
Thus a general physical state $|\psi\rangle$ is a combination of a physical state $|\chi\rangle$ of the form (\ref{basis}) with $\lambda=0$ and a physical spurious state $|s\rangle$.
It remains to show: 

\begin{lemma}\label{lem:phys}
    Let $|\chi\rangle$ be a physical state of the form (\ref{basis}) with $\lambda=0$. Then it follows that $|\chi\rangle\in\mathcal{T}$.
\end{lemma}

\begin{proof}
    Fix a state $|\varphi\rangle\in\mathcal{T}$ and denote the $\hu(1)$ Verma module with highest weight state $|\varphi\rangle$ by $V^X_\varphi$, such that $|\chi\rangle\in V^X_\varphi$. Call the corresponding $c=1$ Virasoro Verma module $V^\mathrm{Vir^X}_\varphi$, and let $L_0|\varphi\rangle=h_\varphi|\varphi\rangle$. If $h_\varphi\neq 0$, the module $V^\mathrm{Vir^X}_\varphi$ is irreducible, and it is well known that it is isomorphic as a graded vector space to $V^X_\varphi$ (with $L_0$-grading in both cases). Since $\widehat{V}_+^{p^+,p^-}$ is a representation of $\mathrm{Vir}^c\oplus\mathrm{Vir}^X$ any highest weight state must simultaneously be a $\mathrm{Vir}^c$ and a $\mathrm{Vir}^X$ highest weight state, so any physical state must in particular be annihilated by all operators $L^X_n$ with $n>0$. It follows that if $h_\varphi\neq 0$, the only state in $V^X_\varphi$ that may be physical is $|\varphi\rangle$, so $|\chi\rangle\in\mathcal{T}$. If $h_\varphi=0$ we must do a bit more work. Recall that $h_\varphi=\frac{(p^+-p^-+r)^2}{2(1-k)}$ where $r\in\mathbb{Z}$, and if $|\varphi\rangle$ is at level $N$ we have the constraint $r\geq -N$. In fact there is only one state in $\widehat{V}_{+(N)}^{p^+,p^-}$ with $r=-N$ namely $(J^+_{-1})^N|p^+,p^-\rangle$, the state on the right border in the $(N+1)$'st row in the following figure.
\begin{center}
\begin{picture}(270,180)
    \put(0,0)           {\includegraphics{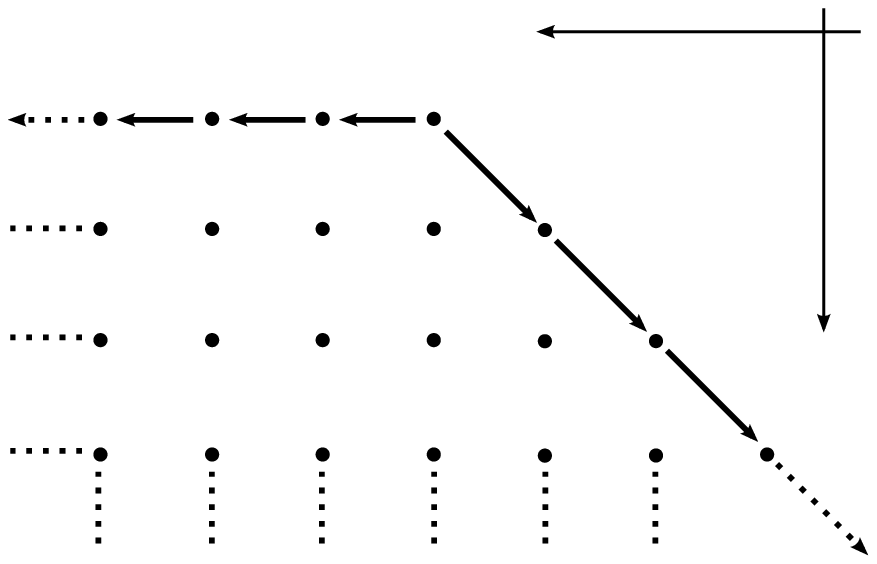}}
    \put(150,160)       {$J_0\ (r)$}
    \put(240,65)        {$L_0\ (N)$}
    \put(105,135)       {\small $J^-_0$}
    \put(73,135)        {\small $J^-_0$}
    \put(40,135)        {\small $J^-_0$}
    \put(145,110)       {\small $J^+_{-1}$}
    \put(172,82)        {\small $J^+_{-1}$}
    \put(205,50)        {\small $J^+_{-1}$}
    \put(128,130)       {\scriptsize $|p^+,p^-\rangle$}
\end{picture}
\end{center}
\end{proof}
A simple calculation shows that all right-border states 
	lie in $\mathcal{T}$, and thus have positive norm.
Now, assume $|\chi\rangle$ is a physical state of the form (\ref{basis}) at level $N$ with $h_\varphi=0$. This implies $p^+-p^-=P$, i.e. $p^-=p^+-P$, for some integer $P\leq N$.  The on-shell condition reads
    $$-\frac{2p^+(p^+-P+\frac{1}{2})}{k}\leq 1-N$$
implying
    $$P\geq p^++\frac{1}{2}+\frac{l}{2p^+}(N-1).$$
The conditions $p^+>0$, $2p^+/|k|<1$ then gives $P>N-1$, so $P=N$. Since the right-border states lie in $\mathcal{T}$ this finishes the proof of the lemma.\endofproof

There is a one-parameter family of automorphisms $\{\theta_s\}_{s\in\mathbb{Z}}$ of $\hnw$ called {\em spectral flow}, acting on our basis of choice as
\begin{eqnarray}
	\theta_s: J^{\pm}_n &\mapsto & J^\pm_{n\pm s}\label{sf1}\\
	\theta_s: J_n & \mapsto & J_n\label{sf2}\\
	\theta_s: T_n & \mapsto & T_n-is\frac{k}{2}\delta_{n,0}.\label{sf3}
\end{eqnarray}
Via the Sugawara construction $\theta_s$ induces automorphisms of the Virasoro algebra acting on the Sugawara generators as
\begin{equation}
	\theta_s: L_n\mapsto L_n -is J_n + is\frac{2}{k}T_n + \frac{s^2}{2}\delta_{n,0}
\end{equation}
If $\widehat{V}$ is a $\hnw$ module where $\rho_V$ is the representation morphism (i.e. we write the action of $X\in\hnw$ on $v\in\widehat{V}$ as $\rho_V(X)v$), denote by ${}^{(s)}\widehat{V}$ the $\hnw$ module with representation morphism $\rho^{(s)}_V:=\rho_V\circ\theta_s$.
We will be interested in extending the results obtained so far for ${}^{(s)}\widehat{V}^{p^+,p^-}_{\pm}$, ${}^{(s)}\widehat{V}^{0,p^-}_\alpha$, and ${}^{(s)}\widehat{V}^{0,0}$, which for brevity will be referred to as ``flowed representations''.

We first make two observations regarding the flowed representations.
\begin{remark}~\\[-1cm]
\begin{itemize}
	\item[(i)] If $T_0$ has eigenvalue $ip^+$ on the module $V$, the $T_0$ eigenvalue on ${}^{(s)}V$ will be $i(p^+-s\frac{k}{2})$. Thus even with the constraint $p^+\in (0,-\frac{k}{2})$ for $\widehat{V}^{p^+,p^-}_+$ and the analogous constraint on $\widehat{V}^{p^+,p^-}_-$, including all flowed representations indicated above the spectrum of $T_0$ becomes $i\mathbb{R}$ (note that to get the eigenvalues $s\frac{k}{2}$ we must include the flowed representations ${}^{(s)}\widehat{V}^{0,p^-}_\alpha$ or ${}^{(s)}\widehat{V}^{0,0}$, however only the former lead to a unitary space of physical states). This is the spectrum obtained by light cone quantization \cite{Cheung}.
	\item[(ii)] We have ${}^{(-1)}\widehat{V}^{p^+,p^-}_+\cong\widehat{V}^{p^++\frac{k}{2},p^-}_-$. This follows from the fact that $\widehat{V}^{p^+,p^-}_\pm$ are Verma modules, and therefore determined by their highest resp. lowest weight vectors. The property $\theta_s\circ\theta_t=\theta_{s+t}$ implies that concentrating on $\widehat{V}^{p^+,p^-}_+$ ($\widehat{V}^{p^+,p^-}_-$) it is enough to consider flowed representations with $s>0$ ($s<0$).
\end{itemize}
\end{remark}

\begin{lemma}
\label{lem:flowbasis}

	Fix $s\in\mathbb{N}$. States of the form (\ref{basis}) in ${}^{(s)}\widehat{V}^{p^+,p^-}_+$ for all levels $N$ span ${}^{(s)}\widehat{V}^{p^+,p^-}_+$.
\end{lemma}
\begin{proof}
	Since $\theta_s(X_n)=X_n$ for $n\neq 0$, a state in $\widehat{V}^{p^+,p^-}$ is $\hu(1)$ primary iff the same vector is also $\hu(1)$ primary in the flowed representation. By rewriting $\theta_s(L_n)$ in terms of $X_n$ and $Y_n$ and using $[X_m,Y_n]=0$, so $Y_{-n}|\psi\rangle$ is $\hu(1)$ primary if $|\psi\rangle$ is, the Lemma follows immediately since the underlying vector space of $\widehat{V}^{p^+,p^-}_+$ coincides with that of ${}^{(s)}\widehat{V}^{p^+,p^-}_+$.\endofproof
\end{proof}
Next, observe that $\theta_s$ preserves the Hermiticity properties of our generators. Since ${}^{(s)}\widehat{V}^{p^+,p^-}_+$ has a cyclic vector, namely the highest weight vector of $\widehat{V}^{p^+,p^-}_+$, we conclude that the inner products on $\widehat{V}^{p^+,p^-}_+$ and ${}^{(s)}\widehat{V}^{p^+,p^-}$ coincide. If we can show that a state of the form (\ref{basis}) in ${}^{(s)}\widehat{V}^{p^+,p^-}_+$ with $\lambda=0$ lies in $\mathcal{T}$ we have managed to show a no-ghost theorem for the flowed representation.

If we can show an analogue of Lemma \ref{lem:phys} for ${}^{(s)}\widehat{V}^{p^+,p^-}_+$ then we have managed to show a no-ghost theorem for the flowed representation.
\begin{lemma}\label{lem:flowphys}
    Let ${}^{(s)}\mathcal{T}\subset{}^{(s)}\widehat{V}^{p^+,p^-}_+$, $s\in\mathbb{N}$, be the subspace spanned by vectors annihilated by all $X_n$ and $L_n$ for $n>0$, i.e. all vectors in $\widehat{V}^{p^+,p^-}_+$ annihilated by $\theta_s(X_n)$ and $\theta_s(L_n)$ for $n>0$, and let $|\chi\rangle\in{}^{(s)}\widehat{V}^{p^+,p^-}_+$ be a physical state of the form (\ref{basis}) with $\lambda=0$. Then $|\chi\rangle\in{}^{(s)}\mathcal{T}$.
\end{lemma}
\begin{proof}
    The strategy of the proof is equivalent to the proof of Lemma \ref{lem:phys}, and we need to investigate the situation where the physical state $|\chi\rangle$ is a descendent of  $|\varphi\rangle$ with $\mathrm{Vir}^X$ weight $h_\varphi=0$. In ${}^{(s)}\widehat{V}^{p^+,p^-}_+$ we have
    $$h_\varphi=\frac{(p^+-p^--s\frac{k}{2}+r)^2}{2(1-k)},$$
where again  $r\geq -N$ if $|\varphi\rangle$ is of level $N$. If $h_\varphi=0$ we have $p^-=p^+-s\frac{k}{2}+r$.
The $L_0$-eigenvalues of states in ${}^{(s)}\widehat{V}^{p^+,p^-}_+$ are of the form
    $$h=-\frac{2p^+(p^-+\frac{1}{2})}{k}+s(p^--r)-s\frac{2p^+}{k}+\frac{s^2}{2}+N,$$
and the mass-shell condition states that $h\leq 1$. Inserting the expression for $p^-$ in the mass-shell condition we get
\begin{equation}
    -\frac{2p^+}{k}\left(p^+-s\frac{k}{2}+\frac{1}{2}+r\right)+s\left(p^+-s\frac{k}{2}-\frac{2p^+}{k}+\frac{s}{2}\right)\leq 1-N.
\end{equation}
Discarding strictly positive terms on the left hand side we get the condition $-\frac{2p^+}{k}r< 1-N$, and since $-\frac{2p^+}{k}\in (0,1)$ we conclude that $r=-N$, so it is again a ``right-border'' state. Such a state is also in ${}^{(s)}\widehat{V}^{p^+,p^-}_+$ annihilated by all $X_n$ and $L_n$ such that $n>0$, and it follows that $|\chi\rangle\in{}^{(s)}\mathcal{T}$.\endofproof
\end{proof}

In section \ref{sec:lowNunitarity} we saw that unitarity for $\widehat{V}^{0,p^-}_\alpha$ follows immediately from the on-shell condition. We now turn to the corresponding flowed representations ${}^{(s)}\widehat{V}^{0,p^-}_\alpha$. Note first that Lemma \ref{lem:basis} also holds in the representations $\widehat{V}^{0,p^-}_\alpha$ (the relevant $L^c_0$-eigenvalues $h^c$ are trivially positive), and hence also in the flowed representations ${}^{(s)}\widehat{V}^{0,p^-}_\alpha$ respectively ${}^{(s)}\widehat{V}^{0,0}$ by the same argument as in Lemma \ref{lem:flowbasis}. We are thus left to show the analogue of Lemma \ref{lem:flowphys} for these representations, which again reduces to analyzing $\mathrm{Vir}^X$ descendents of weight $0$ primaries $|\varphi\rangle$. In the representation ${}^{(s)}\widehat{V}^{0,p^-}_\alpha$ we get $h_\varphi=-\frac{(p^--n+s\frac{k}{2})^2}{2(1-k)}$, and $h_\varphi=0$ implies $p^--n=-s\frac{k}{2}$. A corresponding $L_0$-eigenvalue then takes the form
$$h=-\frac{\alpha^2}{k} + \frac{s^2}{2}\left(1-k\right)+N,$$
and since $1-k>0$ the on-shell condition implies $N=0$, so all physical states lie in the horisontal submodule $V^{0,p^-}_\alpha$ which is unitary.

We summarize the results above in the following
\begin{theorem}
    The space of physical states for the bosonic string on the Nappi-Witten group with level $k$ is unitary if $k<0$ in the following cases
    \begin{itemize}
        \item the space ${}^{(s)}\widehat{V}_+^{p^+,p^-}$ when $s\in\mathbb{N}_0$, $p^+\in(0,-\frac{k}{2})$, $p^-\in\mathbb{R}$
        \item the space ${}^{(s)}\widehat{V}_-^{p^+,p^-}$ when $s\in-\mathbb{N}_0$, $p^+\in(0,\frac{k}{2})$, $p^-\in\mathbb{R}$
        \item the space ${}^{(s)}\widehat{V}_\alpha^{0,p^-}$ when $s\in\mathbb{Z}$, $\alpha,p^-\in\mathbb{R}$
    \end{itemize}
\end{theorem}

\section{Discussion}
In this paper, we presented the proof of a no-ghost theorem for the bosonic string
in Nappi-Witten spacetime. The proof involves two parts.  The first
part consists of showing that the states in a class of representations of
the coset $NW/U(1)$ have non-negative norm, where the $U(1)$ subgroup
corresponds to a timelike direction in the Nappi-Witten
spacetime. The proof of this part follows closely the discussion
of the unitarity of the $SO(2,1)/U(1)$ modules by Dixon \it et al
\rm \cite{Dixon}. In the second part,  physical states satisfying
the Virasoro constrains are shown to lie within the coset modules
if we assume the total Hilbert space of the $NW$ WZW model decomposes completely in a certain set of representations of the $NW$ current algebra.
The proof in this part follows a more direct and clear proof  of
unitarity of the $SU(1,1)$ bosonic string theory given in
\cite{Hwang2} than the one given in \cite{Hwang}.

A crucial ingredient for the first part is that every state in the
current algebra representation can be expanded as a sum
of terms, each of which is obtained by acting a product of
negative mode operators of the timelike generator $J_{-n}-T_{-n}$
($n>0$) on a state annihilated by any $J_{n}-T_{n}$ ($n>0$). One nice property of the
Nappi-Witten group is that we can construct this expansion
explicitly.

In addition to establishing a unitary spectrum one would also like to identify physical configurations corresponding to these states. In particular we expect to find analogues of the long string states in $AdS_3$ \cite{MalOog1} for $p^+\in \frac{k}{2}\mathbb{Z}$, which with our proposal for the representation content would correspond to states in ${}^{(s)}V^{0,p^-}_\alpha$. Another obvious continuation is the construction of modular invariant partition functions. We refer to a future publication for both topics \cite{CCFFH2}.

\section*{Acknowledgement}

This work is initiated and supported in parts by the Research Links Programme of Swedish Research Council under contract No.~348-2008-6049.   G.~C. and Z.~F. are supported in parts by NSFC Grants No.~10535010 and No.~10775068 as well as by 973 National Major State Basic Research and Development of China Grant No.~2007CB815004. E.~C acknowledges support from the National Science Foundation of China under grant No.~0204131361   and Startup grants from Nanjing University as well as 985 Grant No.~020422420100 from the Chinese Government.

\end{document}